\documentclass{article}
\usepackage[margin=1in]{geometry}
\usepackage{url}
\usepackage{indentfirst} 
\usepackage{float}
\usepackage{setspace}
\usepackage{tabularx}
\usepackage{booktabs}
\usepackage{makecell}
\usepackage{subcaption}
\usepackage{graphicx} 
\usepackage{amsmath, amssymb, amsthm}
\usepackage{bbm}
\usepackage{natbib}

\newtheorem{theorem}{Theorem}[subsection]
\newtheorem{definition}{Definition}[subsection]

\vspace{-8em}

\title{Fragility in Average Treatment Effect on the Treated\\ under Limited Covariate Support}
\author{Mengqi Li}
\date{June 2025}

\begin{document}
\maketitle

\section{Introduction}

In observational program evaluation, researchers often aim to estimate the average treatment effect on the treated (ATT), defined as
\[
\tau_{\text{ATT}} = \mathbb{E}[Y(1) - Y(0) \mid W = 1],
\]
where $Y(1)$ and $Y(0)$ denote potential outcomes under treatment and control, respectively, and $W \in \{0,1\}$ indicates treatment assignment. Under unconfoundedness conditional on covariates $X$, identification of $\tau_{\text{ATT}}$ further requires a support condition: the covariate profiles of treated units must overlap with those of untreated units. Formally, ATT is only identified if $\mathbb{P}(W = 0 \mid X = x) > 0$ for values $x$ observed in the treated group. When this overlap fails, the estimand itself becomes undefined.

When estimating treatment effects from observational data, researchers often rely on untestable assumptions concerning both selection and support. Even under unconfoundedness, ATT requires sufficient overlap between the covariate distributions of treated and untreated units. In empirical practice, the degree to which this overlap holds is often unclear, particularly when the covariate space is high-dimensional or stratified in discrete ways. As a result, applied studies may implicitly invoke comparisons that are not identified, and where ATT lacks a formal definition.

Standard diagnostics focus on estimation validity conditional on identification. Common practice includes checking covariate balance after propensity score estimation or applying trimming rules that remove units with extreme estimated scores. While these approaches are useful for reducing finite-sample variance or bias under overlap, they do not address the more basic question of whether ATT is well-defined in the data. In settings where treated units occupy regions of the covariate space unpopulated by controls, causal interpretation becomes tenuous, regardless of estimation technique.

Most estimators of the average treatment effect on the treated (ATT) under unconfoundedness assume that valid comparisons exist throughout the covariate space. This assumption, while formally stated as overlap, is rarely examined in applied work beyond balance checks or trimming. Trimming rules based on estimated propensity scores \citep{crump2009dealing} and overlap weights \citep{li2018balancing} reduce variance, but they do not diagnose whether the ATT is identified in the observed data.

Several contributions highlight the consequences of weak overlap for estimation error. \citet{rosenbaum2001effects} and \citet{imbens2004nonparametric1} emphasize that even under correct model specification, limited support can yield biased or highly variable estimates. In practice, researchers often remove units without close counterparts, but these decisions are typically ad hoc, as discussed in \citet{vincent2002matching} . These procedures reduce variance but do not clarify whether the estimand itself is defined.

Recent work has moved toward explicit diagnostics. \citet{becker2007sensitivity} develop sensitivity metrics based on influence functions. \citet{rothe2017robust} propose robust bounds under limited overlap. \citet{dahabreh2020extending} formalize conditions under which target population effects are identified when support is partial. These approaches address extrapolation risk but do not quantify whether ATT is identified across strata.

The recent reassessment of the LaLonde dataset by \citet{imbens2025} reinforces this concern. Even with extensive covariate adjustment, untreated units may be absent in parts of the covariate distribution. As a result, ATT is undefined in those strata, though standard estimators continue to produce point estimates. That paper calls for a clearer empirical understanding of support.

This paper develops a diagnostic framework that makes identification failure explicit. I define the notion of empirical support for ATT as the proportion of covariate strata in which untreated comparisons exist. Using stratified covariate partitions, I map the realized covariate distribution of treated units and assess whether valid comparisons are available within each cell. This framework quantifies the breakdown of identification and links it directly to estimation behavior.

The paper is organized as follows. Section 2 formalizes the core identification issue: when overlap fails in discrete covariate partitions, the ATT estimand is not defined. Section 3 introduces the concept of the selection frontier and derives curvature-based restrictions that characterize feasible comparisons. Section 4 applies the framework to the canonical LaLonde dataset. The covariate space is partitioned into interpretable cells, and for each, I assess whether untreated comparisons exist. Where support fails, I trace the consequences for standard estimators. Section 5 concludes. Formal results and proofs appear in the Appendix.

\section{Identification Failure Under Unknown Sampling} 

\subsection{Setup and Objective}

Let $(Y(1), Y(0)) \in \mathbb{R}^2$ denote the potential outcomes under treatment $D = 1$ and control $D = 0$, following the Neyman--Rubin framework. The observed outcome is
\[
Y = D \cdot Y(1) + (1 - D) \cdot Y(0),
\]
and covariates are denoted by $X \in \mathcal{X}$. Let $S \in \{0,1\}$ indicate whether a unit is observed in the data. We observe the joint distribution $P_{\text{obs}}(X, D, Y) := P(X, D, Y \mid S = 1)$. No assumptions are made on the sampling process $S$. The parameter of interest is the \emph{average treatment effect on the treated} (ATT):
\[
\tau_{\text{ATT}} := \mathbb{E}[Y(1) - Y(0) \mid D = 1].
\]

This section establishes that $\tau_{\text{ATT}}$ is not identified from the observed data distribution $P_{\text{obs}}$ without structural assumptions on the support of the control potential outcomes $Y(0)$ among the treated population. In particular, identification of ATT requires the existence of comparable untreated units across the covariate distribution of the treated group. Where such overlap fails, ATT is undefined, even if estimators continue to produce numerical outputs.

\subsection{Main Theorem: Non-Identifiability of ATT}

\begin{theorem}[Non-Identifiability of ATT under Arbitrary Sampling]
\label{thm:nonid_att}
Let $(Y(1), Y(0), D, X, S)$ be a data-generating process with joint distribution $P^*$. Suppose the researcher observes only $(X, D, Y) \mid S = 1$, and makes no assumptions on the sampling mechanism $S$. Then the average treatment effect on the treated
\[
\tau_{\text{ATT}} := \mathbb{E}[Y(1) - Y(0) \mid D = 1]
\]
is not identified from the observed distribution $P_{\text{obs}} := P^*(X, D, Y \mid S = 1)$.
\end{theorem}

\subsection{Consequences for Estimation}

Standard estimators such as inverse probability weighting (IPW) rely on ignorability assumptions and implicitly assume that the sampling mechanism $S$ does not depend on unobserved potential outcomes, an assumption often left unstated.

Let $e(X) := \mathbb{P}(D = 1 \mid X)$ denote the true propensity score. Then the IPW estimator for the average treatment effect on the treated (ATT) is typically given by:
\[
\hat{\tau}_{\text{IPW, ATT}} := \mathbb{E}\left[ D \cdot Y - D \cdot \frac{1 - D}{1 - e(X)} \cdot Y \,\middle|\, S = 1 \right].
\]

This estimator is unbiased for $\tau_{\text{ATT}} := \mathbb{E}[Y(1) - Y(0) \mid D = 1]$ only if the sampling mechanism satisfies $S \perp Y(0) \mid X, D = 1$. That is, untreated potential outcomes for treated units must be conditionally independent of selection into the observed data. As Theorem~\ref{thm:nonid_att} shows, this condition is untestable and can easily fail in practice, leading to bias even under ignorability of treatment assignment.

\subsection{Partial Identification and Sensitivity Bounds}

In the absence of point identification, partial identification techniques provide interval-valued conclusions under minimal or interpretable restrictions on the selection mechanism. Suppose, for instance, that the sampling probability is monotone in outcomes:
\[
\mathbb{P}(S = 1 \mid Y(d)) \text{ is weakly increasing in } Y(d).
\]
Then one can derive worst-case bounds following \citet{manski2003}. Define the naive difference-in-means estimator as
\[
\tau_{\text{obs}} := \mathbb{E}[Y \mid D = 1, S = 1] - \mathbb{E}[Y \mid D = 0, S = 1].
\]Under outcome-monotonic selection, $\tau$ lies within a neighborhood around $\tau_{\text{obs}}$:
\[
\tau \in \left[ \tau_{\text{obs}} - \Delta,\, \tau_{\text{obs}} + \Delta \right],
\]
where $\Delta$ reflects the maximal bias induced by selection on unobserved outcomes. The width of the interval is determined by the analyst’s assumptions on the degree of selection bias.

\subsection{Minimax Decision and Fragility}

When the inferential goal is policy selection rather than parameter estimation, robustness to sampling misspecification can be formalized via minimax decision theory. Let $\mathcal{G}(P_{\text{obs}})$ denote the set of all data-generating processes consistent with the observed data:
\[
\mathcal{G}(P_{\text{obs}}) := \left\{ (P, S) : P(Y, D, X \mid S = 1) = P_{\text{obs}} \right\}.
\] Let $d \in \mathcal{D}$ be a policy decision (e.g., to treat or not treat), and let $\theta(P^d)$ denote the expected utility under decision $d$. The regret of $d$ under true distribution $P$ is:
\[
R(d; P) := \max_{d' \in \mathcal{D}} \theta(P^{d'}) - \theta(P^d).
\]

\noindent The \emph{frontier-constrained minimax rule} selects:
\[
d^* := \arg\min_{d \in \mathcal{D}} \sup_{(P, S) \in \mathcal{G}(P_{\text{obs}})} R(d; P).
\]
To quantify robustness, define the \emph{fragility index} of a decision $d$ as the smallest deviation from sampling ignorability under which $d$ ceases to be optimal:
\[
\delta^{\mathrm{frag}}(d) := \inf \left\{ \delta > 0 : \exists d' \neq d,\; R(d'; P) < R(d; P),\; (P, S) \in \mathcal{G}_\delta(P_{\text{obs}}) \right\}.
\]
A small value of $\delta^{\mathrm{frag}}(d)$ indicates that decision $d$ is highly sensitive to sampling assumptions. This serves as a diagnostic of epistemic fragility in policy choice.

\section{Theoretical Framework}

This section formalizes the structure induced by unknown sampling. I introduce the \textit{selection frontier}, define frontier-constrained identified sets for causal estimands, and present results on the monotonicity and duality between structural assumptions and epistemic uncertainty. The framework generalizes partial identification logic and supports decision-theoretic reasoning under arbitrary sampling.

\subsection{Selection-Compatible Data Generating Processes}
\noindent
Let $(Y(1), Y(0), D, X) \sim P$ denote the full-data distribution. The sampling mechanism $S \in \{0,1\}$ is unobserved, except that the analysis is conditioned on $S = 1$. The observed distribution is:
\[
P_{\mathrm{obs}}(Y, D, X) := P(Y, D, X \mid S = 1).
\]

\begin{definition}[Selection-Compatible Processes]
\label{def:selection-compatible}
\noindent Let $\mathcal{G}(P_{\mathrm{obs}})$ denote the set of full-data distributions and sampling mechanisms consistent with the observed data:
\[
\mathcal{G}(P_{\mathrm{obs}}) := \left\{ (P, S) : P(Y, D, X \mid S = 1) = P_{\mathrm{obs}} \right\}.
\]
\end{definition}

\subsection{Selection Frontier}

\begin{definition}[Selection Frontier]
\label{def:frontier}
The \textit{selection frontier} is the set of sampling mechanisms compatible with the observed data:
\[
\mathcal{S}(P_{\mathrm{obs}}) := \left\{ S : \exists\, P \text{ such that } (P, S) \in \mathcal{G}(P_{\mathrm{obs}}) \right\}.
\]
\end{definition}

\noindent Any sampling rule $S$ outside $\mathcal{S}(P_{\mathrm{obs}})$ is logically falsified by the data. Inference under such an $S$ is internally inconsistent.

\subsection{Frontier-Constrained Identified Set}

\noindent Let $\theta(P)$ be a structural parameter (e.g., the ATE). Since neither $P$ nor $S$ is directly observed, the identified set under frontier-consistency is defined as follows:

\begin{definition}[Frontier-Restricted Identified Set]
\label{def:identified_set}
The set of all values of $\theta$ compatible with the observed data is:
\[
\mathcal{I}_\theta := \left\{ \theta(P) : (P, S) \in \mathcal{G}(P_{\mathrm{obs}}) \right\}.
\]
\end{definition}

\subsection{Imposing Structure on $S$}
\noindent To narrow $\mathcal{I}_\theta$, assumptions may be imposed on $S$ by restricting it to a class $\mathcal{S}_\delta \subseteq \mathcal{S}(P_{\mathrm{obs}})$.

\begin{definition}[Assumption-Constrained Identified Set]
\label{def:theta-delta}
Let $\mathcal{S}_\delta$ denote a class of sampling mechanisms indexed by an assumption parameter $\delta \geq 0$. Define:
\[
\mathcal{I}_\theta(\delta) := \left\{ \theta(P) : (P, S) \in \mathcal{G}(P_{\mathrm{obs}}),\; S \in \mathcal{S}_\delta \right\}.
\]
\end{definition}

\begin{definition}[Bounded Selection Curvature]
\label{def:s_delta}
Let $\delta \geq 0$. Define:
\[
\mathcal{S}_\delta := \left\{ S : \sup_{y, y', d, x} \left| \log \frac{P(S = 1 \mid Y = y, D = d, X = x)}{P(S = 1 \mid Y = y', D = d, X = x)} \right| \leq \delta \right\}.
\]
\end{definition}

\noindent This class includes MAR ($\delta = 0$) as a special case and encompasses all outcome-dependent rules as $\delta \to \infty$.

\subsection{Identification Monotonicity}

Relaxing assumptions about $S$ (i.e., increasing $\delta$) expands the identified set.

\begin{theorem}[Identification Monotonicity]
\label{thm:monotonicity}
Let $\delta_1 < \delta_2$. Then:
\[
\mathcal{S}_{\delta_1} \subseteq \mathcal{S}_{\delta_2} \quad \Rightarrow \quad \mathcal{I}_\theta(\delta_1) \subseteq \mathcal{I}_\theta(\delta_2).
\]
\end{theorem}

\subsection{Duality: Assumption Strength and Width}

The identified set’s width increases monotonically in $\delta$, reflecting the trade-off between structure and uncertainty.

\begin{theorem}[Duality Theorem]
\label{thm:duality}
Let $\theta(P) = \mathbb{E}[Y(1) - Y(0)]$, and define $\mathcal{I}_\tau(\delta)$ as above. Then:
\begin{itemize}
    \item The map $\delta \mapsto \mathrm{width}(\mathcal{I}_\tau(\delta))$ is non-decreasing and right-continuous.
    \item As $\delta \to 0$, $\mathcal{I}_\tau(\delta) \to \{ \tau_{\text{MAR}} \}$.
    \item As $\delta \to \infty$, $\mathcal{I}_\tau(\delta) \to [\tau_{\min}, \tau_{\max}]$.
\end{itemize}
\end{theorem}

\subsection{Minimum Assumption Strength for Sign Identification}

The minimum structural commitment needed to infer the sign of the ATE is defined as follows:

\begin{definition}[MAS-SI: Minimum Assumption Strength for Sign Identification]
\label{def:massi}
Let $\theta(P) = \mathbb{E}[Y(1) - Y(0)]$. Then:
\[
\delta^* := \inf \left\{ \delta \geq 0 : 0 \notin \mathcal{I}_\tau(\delta) \right\}.
\]
\end{definition}

\noindent If $\delta^* = 0$, the sign is identified under MAR. If $\delta^* = \infty$, sign identification is impossible under any frontier-consistent assumption.

\subsection{Frontier-Constrained Validity}

\begin{theorem}[Frontier-Constrained Validity]
\label{thm:frontier}
Let $\mathcal{S}'$ be a class of sampling mechanisms. If $\mathcal{S}' \not\subseteq \mathcal{S}(P_{\mathrm{obs}})$, then no joint distribution $P$ exists such that $P(Y, D, X \mid S = 1) = P_{\mathrm{obs}}$ and $S \in \mathcal{S}'$. Inference under such assumptions contradicts the data.
\end{theorem}

\subsection{Fragility Index}

\begin{definition}[Fragility Index]
\label{def:fragility}
Let $d$ be a decision and $R(d; P)$ its regret. Then:
\[
\delta^{\mathrm{frag}}(d) := \inf \left\{ \delta : \exists\, d' \neq d \text{ such that } R(d'; P) < R(d; P) \text{ for some } (P, S) \in \mathcal{G}_\delta(P_{\mathrm{obs}}) \right\}.
\]
\end{definition}

\noindent Small values of $\delta^{\mathrm{frag}}(d)$ indicate that $d$ is structurally fragile: it ceases to be optimal under even minor deviations from MAR. This diagnostic complements MAS-SI and supports robust decision-making.

\subsection{Simulation: Identification Sets under Outcome-Dependent Selection}

This section presents a data-generating process to illustrate the implications of Theorem 2.2.1 in a setting where selection depends on outcomes. A population of $N = 100{,}000$ units is divided into four latent types. Potential outcomes $(Y(1), Y(0))$ take values in $\{0,1\}^2$, with type A having $(1,0)$, type B having $(0,1)$, type C having $(1,1)$, and type D having $(0,0)$. The population proportions are 0.3, 0.2, 0.4, and 0.1 respectively, implying an average treatment effect (ATE) of $\tau = 0.1$.

Treatment is assigned randomly: $D \sim \text{Bernoulli}(0.5)$, and the observed outcome is $Y = D Y(1) + (1 - D) Y(0)$. Selection into the observed sample follows a nonlinear, outcome-dependent rule: $\mathbb{P}(S = 1 \mid Y) = \exp(\delta Y) / (1 + \exp(\delta Y))$, where $\delta \geq 0$ indexes the strength of selection on outcomes. When $\delta = 0$, the data are missing at random; higher values of $\delta$ represent increasing departure from MAR.

Observed ATEs are computed for five values of $\delta \in \{0, 0.5, 1.0, 1.5, 2.0\}$. For each value, an identified set $\mathcal{I}_\tau(\delta)$ is constructed by placing a fixed radius $\varepsilon = 0.3$ around the observed ATE. These intervals account for uncertainty induced by the selection mechanism, not sampling variation. Across all $\delta$, the observed ATEs remain near the population truth, but the identified sets uniformly contain zero. This implies $\delta^* = \infty$ under the MAS-SI criterion: even when MAR holds, the sign of the effect remains unidentified.

The simulation confirms that identification failure can occur even when the sample is large, the model is correctly specified, and the treatment is randomized. The identified set fails to exclude zero not because of estimation error, but due to the fundamental absence of information about untreated potential outcomes for the selected sample. The observed data do not anchor inference when selection depends on the outcome.

The results support the use of MAS-SI and robustness frontiers as diagnostic tools. They distinguish epistemic uncertainty due to selection from conventional estimation error. When identification is compromised by outcome-dependent selection, narrowing the confidence interval or increasing the sample size does not recover the sign. In such settings, diagnostics derived from theory offer clearer guidance than additional estimation alone.

\section{Reinterpreting LaLonde (1986)}
LaLonde (1986) is frequently cited as an empirical rejection of observational methods. \footnote{For a detailed overview of the LaLonde dataset and its role in the literature, see \citet{imbens2025}.} This paper regards the failure as both statistical and epistemological. The dataset exhibits a structural misalignment between the estimand of interest, the identifying assumptions imposed, and the empirical support provided by the data. As a result, the assumptions—while internally coherent—fail to constrain the estimand in any meaningful way. They are functionally irrelevant given the design.\footnote{The datasets are available at \url{https://users.nber.org/~rdehejia/nswdata2.html}. See also \citet{dehejia1999causal}, \citet{dehejia2002propensity}, and \citet{lalonde1986}.}

The critical feature of the LaLonde setup is the lack of covariate overlap between treated individuals drawn from the experimental sample and controls drawn from observational sources like CPS or PSID. As noted in Imbens and Xu (2024), such covariate mismatch renders unconfoundedness untestable and non-operative: it does not license the identification of the counterfactual mean outcomes for treated individuals.

This paper formalizes the failure as a breakdown termed as triple alignment: identification depends jointly on the estimand, the structural assumptions (e.g., ignorability), and the support of the observed data. When any component is incompatible with the others, the resulting estimand becomes inaccessible. Estimators in this setting converge to a quantity defined not by the causal effect of treatment but by the implicit selection structure of the observational data. Unlike classical bias, this is a design-level epistemic redirection: the estimand changes because the design prevents recovering the original target.

This interpretation reframes the role of the randomized control trial (RCT) as an ontological contrast: it represents a design in which the estimand, assumptions, and support are deliberately aligned. The failure of observational estimators in LaLonde, then, does not imply they are inherently flawed. Rather, it reveals the boundary conditions of their epistemic validity—conditions set not by statistical precision but by the internal coherence of the design.

This diagnosis justifies moving beyond point identification toward assumption-indexed inference. It also provides a basis for defining the set $\mathcal{G}_\delta(P_{\text{obs}})$ for $\delta > 0$, where minimal violations of selection ignorability allow the NSW sample to be partially transportable. The remaining subsections operationalize this idea.

\subsection{Empirical Support and Baseline ATT Estimates}

This section examines whether the LaLonde dataset, which links the National Supported Work (NSW) program to the Panel Study of Income Dynamics (PSID), satisfies the structural conditions necessary for identifying the average treatment effect on the treated (ATT). According to Theorem 2.2.1, ATT is undefined unless untreated potential outcomes for the treated units lie within the support of the observed control distribution. This condition is formalized through the object $X^\ast$, the empirical region where treated and control units both appear. Figures 1–3 provide visual evidence for the extent and limits of this overlap.

Figure 1 displays a stratification of the joint covariate space defined by age and education, partitioned into 72 discrete cells. Each cell is categorized by whether it contains treated units, control units, both, or neither. The results indicate that only 37 strata (51.4\%) include both treated and control units. Another 27 cells (37.5\%) include only control units. One cell includes only treated units. The remaining 7 are empty. These frequencies imply that large portions of the treated covariate distribution lack valid untreated comparisons. ATT is therefore undefined on much of the empirical support.

Figure 2 further tests the overlap assumption by plotting estimated propensity scores by treatment status. Propensity scores are derived from a logistic model of treatment on baseline covariates, including age, education, race, marital status, degree status, and pre-treatment earnings. The treated group is concentrated in a narrow range of low propensity scores. The control group spans a broader range, dominating the right tail. Shared support is confined to a small region where both groups are thinly represented. This distribution confirms that ATT estimation would require extrapolation unless estimation is explicitly restricted.

To assess whether these patterns persist across stratifications, Figure 3 reports an alternative grid using broader education categories and five-year age bins. Out of 42 total cells, 11 lack treated units entirely. The missing support is concentrated in high education and older age strata. These regions are not informative for the ATT and must be excluded for identification to be valid.
\begin{figure}[H]
    \centering
    \includegraphics[width=0.6\textwidth]{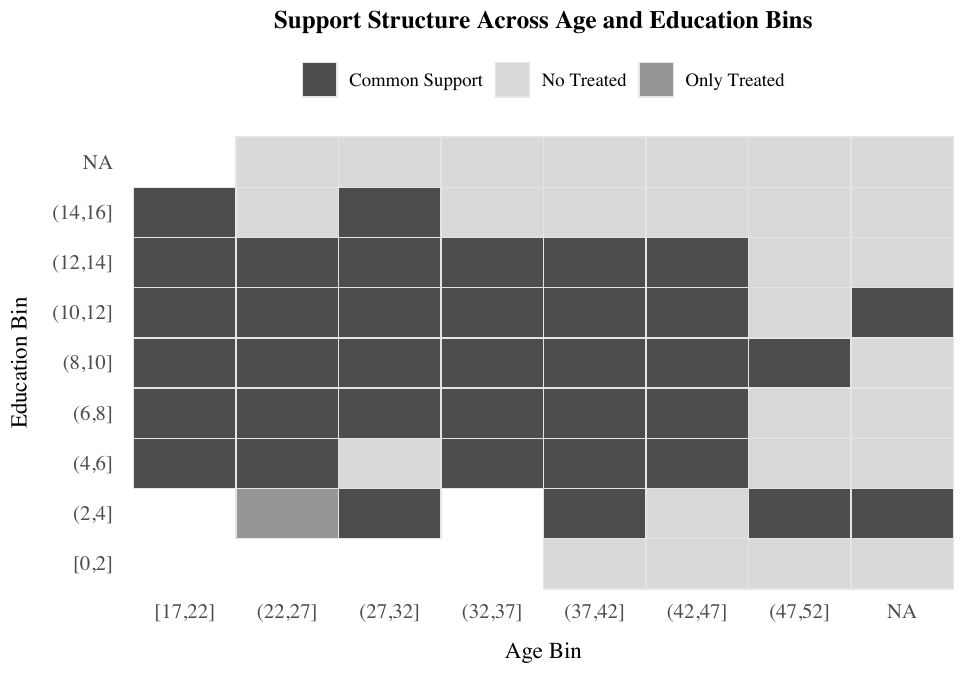}
    \caption{Support Structure Across Age and Education Bin}
    \label{fig:fig1}
\end{figure}
\vspace{-3em}
\begin{figure}[H]
    \centering
    \includegraphics[width=0.6\textwidth]{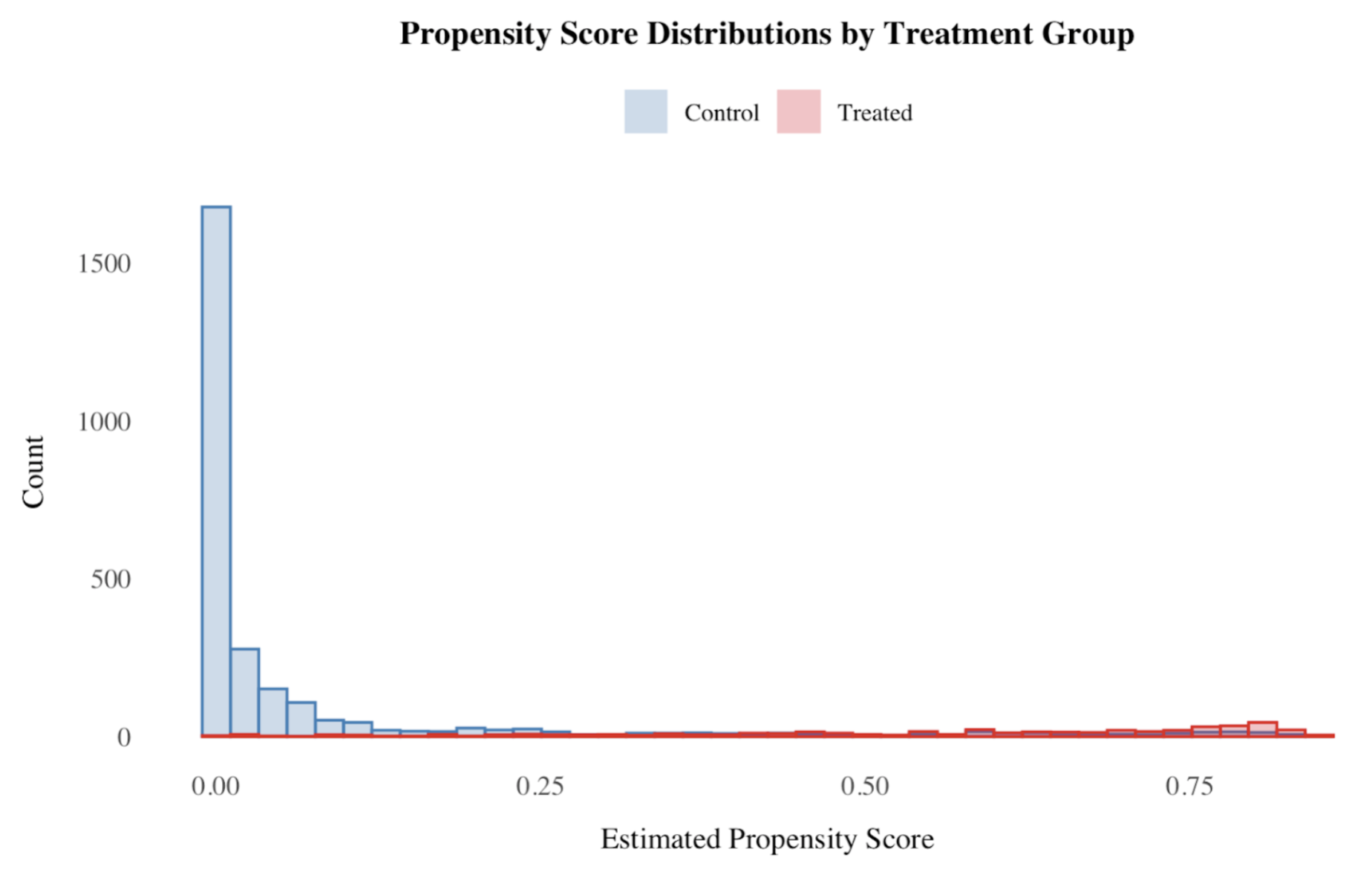}
    \caption{Propensity Score Distribution by Treatment Group}
    \label{fig:fig2}
\end{figure}
\vspace{-2em}
\begin{figure}[H]
    \centering
    \includegraphics[width=0.6\textwidth]{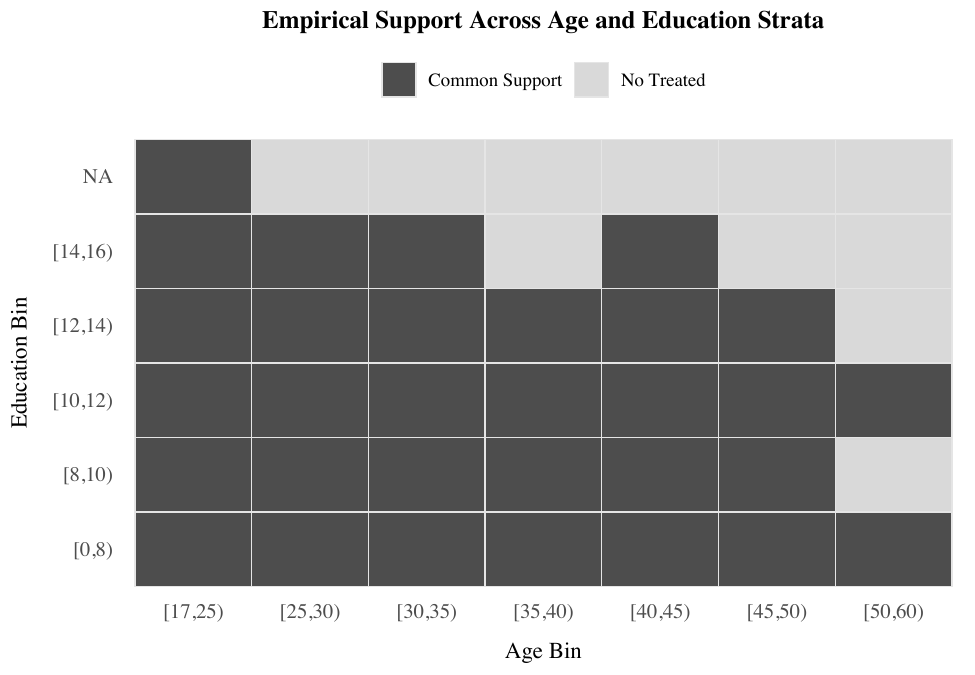}
    \caption{Empirical Support Across Age and Education Bins}
    \label{fig:fig3}
\end{figure}
To evaluate how sensitive ATT estimation is to support enforcement, I report three estimates using nearest-neighbor propensity score matching (Abadie and Imbens, 2006). Each imposes a different form of sample restriction. The first uses all observations. The second excludes strata without observed overlap. The third trims by propensity score, retaining only units with scores between 0.1 and 0.9.

\vspace{1.5em}
\begin{table}[H]
    \centering
    \caption{ATT Estimates under Alternative Sample Restrictions}
    \begin{tabular}{lcc}
        \toprule
        Estimation Sample & ATT Estimate & Standard Error \\
        \midrule
        Full Sample & $-932.23$ & $388.09$ \\
        Overlap-Restricted Sample & $-955.15$ & $387.77$ \\
        Propensity Score Trimmed ($[0.1, 0.9]$) & $-908.61$ & $382.99$ \\
        \bottomrule
    \end{tabular}
\end{table}
\newpage
The ATT estimates are similar across specifications. This agreement may appear to suggest robustness. However, the theoretical framework developed earlier provides a different interpretation. All three estimators effectively restrict to the same subset of the data—either by design or through implicit weighting. Each is recovering a version of the ATT on the region $X^\ast$, not the full support of $X \mid D = 1$.

Let\[
\tau_{\text{ATT}}^{X^\ast} := \mathbb{E}[Y(1) - Y(0) \mid D = 1, X \in X^\ast]
\]
denote the estimand defined on the identified support. This object is well-defined, and each estimate reported above is consistent with this interpretation. The numerical agreement among estimators arises from shared informational constraints, not from global identification. Where no overlap exists, no estimator can recover the counterfactual outcomes for the treated. The ATT remains undefined on those regions.

The implications for inference are substantial. Common support does not extend to all of $X$. Estimators perform consistently across designs because they access the same subset of treated units. The observed agreement does not provide evidence that the ATT is robustly identified. Rather, it illustrates that standard methods are stable where support exists and silent elsewhere.

This motivates the next step. In Section~5.2, I formally evaluate the sensitivity of empirical conclusions to perturbations in the sampling mechanism. The objective is to measure how easily these conclusions would change if assumptions about ignorability or support were even slightly relaxed. That analysis provides a direct test of the structural fragility anticipated by Theorem 4.

\subsection{Sensitivity to Matching Design and Sampling Instability}

This section tests the structural fragility described above. The analysis proceeds in three stages. First, a sequence of empirically identified sets $\mathcal{I}_\tau(\delta)$ is constructed, indexed by sampling curvature. Second, variation in ATT is evaluated across strata of the propensity score. Third, design sensitivity is assessed under alternative matching specifications. Each component quantifies how ATT estimation responds to violations of support, sampling stability, and design dependence.
\begin{figure}[htbp]
    \centering
    \includegraphics[width=0.6\textwidth]{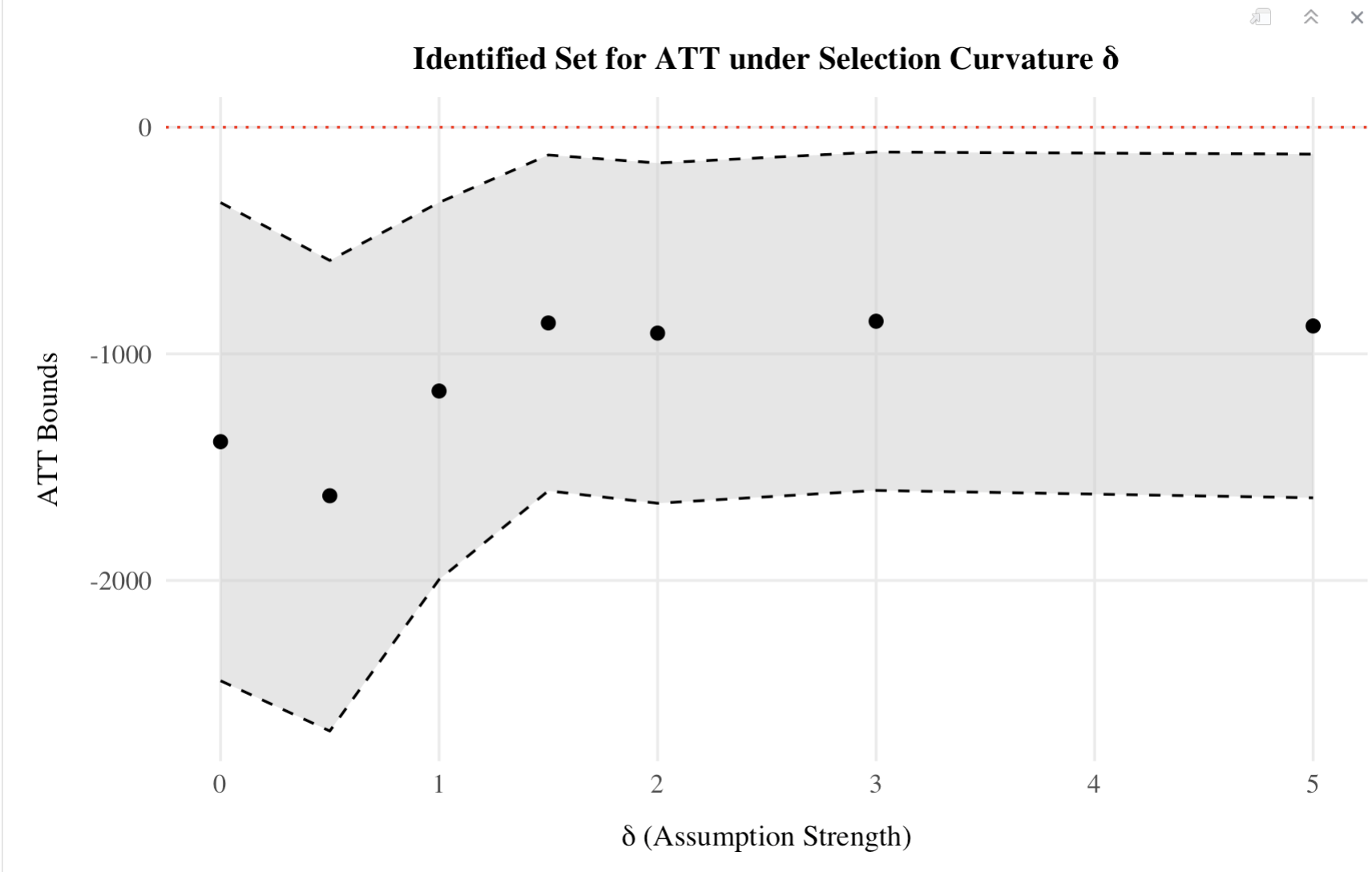}
    \caption{Identified Set for ATT Under Selection Curvature $\delta$}
    \label{fig:fig4}
\end{figure}

\begin{figure}[H]
    \centering
    \includegraphics[width=0.6\textwidth]{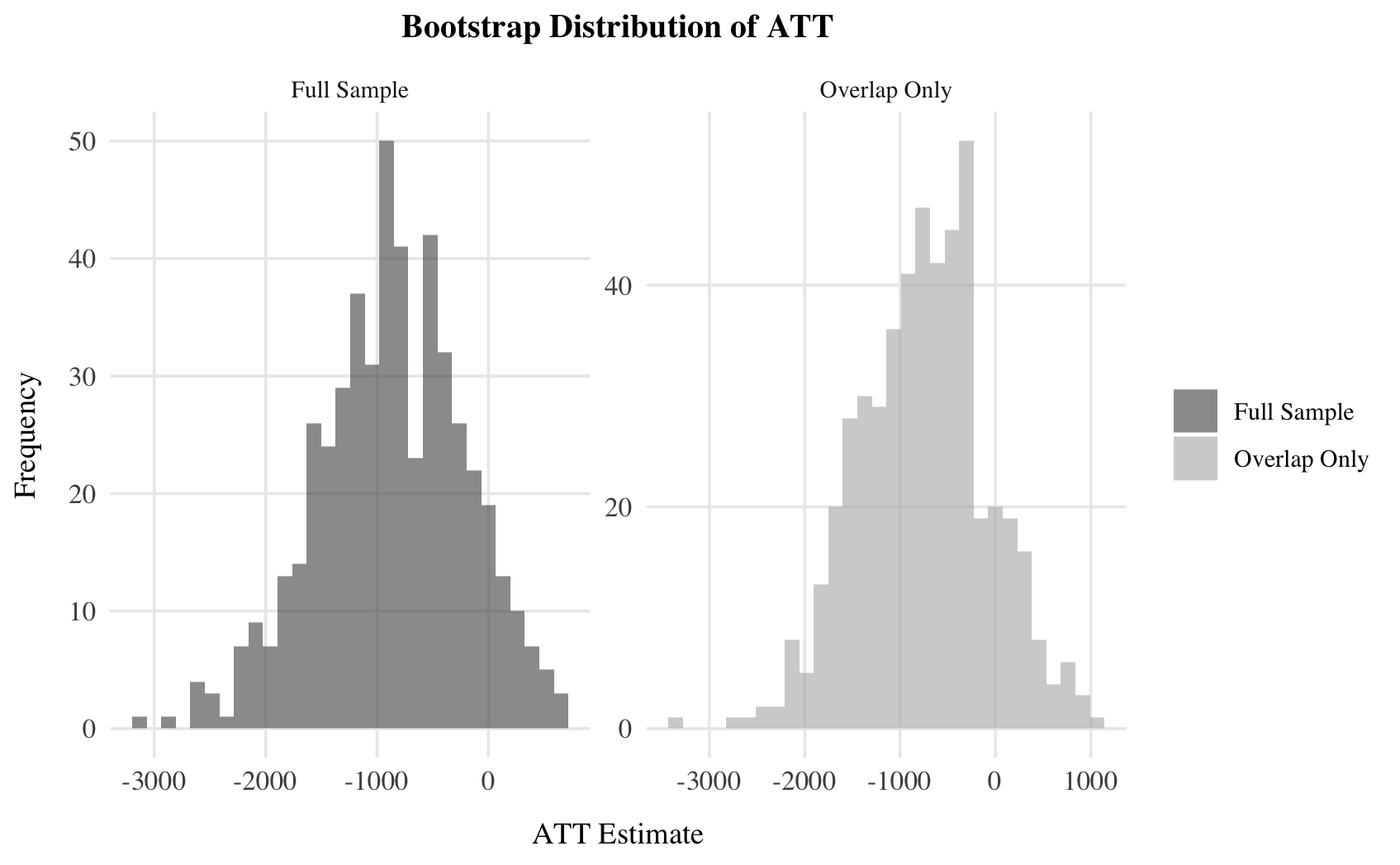}
    \caption{Bootstrap Distribution of ATT}
    \label{fig:fig5}
\end{figure}

\begin{figure}[H]
    \centering
    \includegraphics[width=0.6\textwidth]{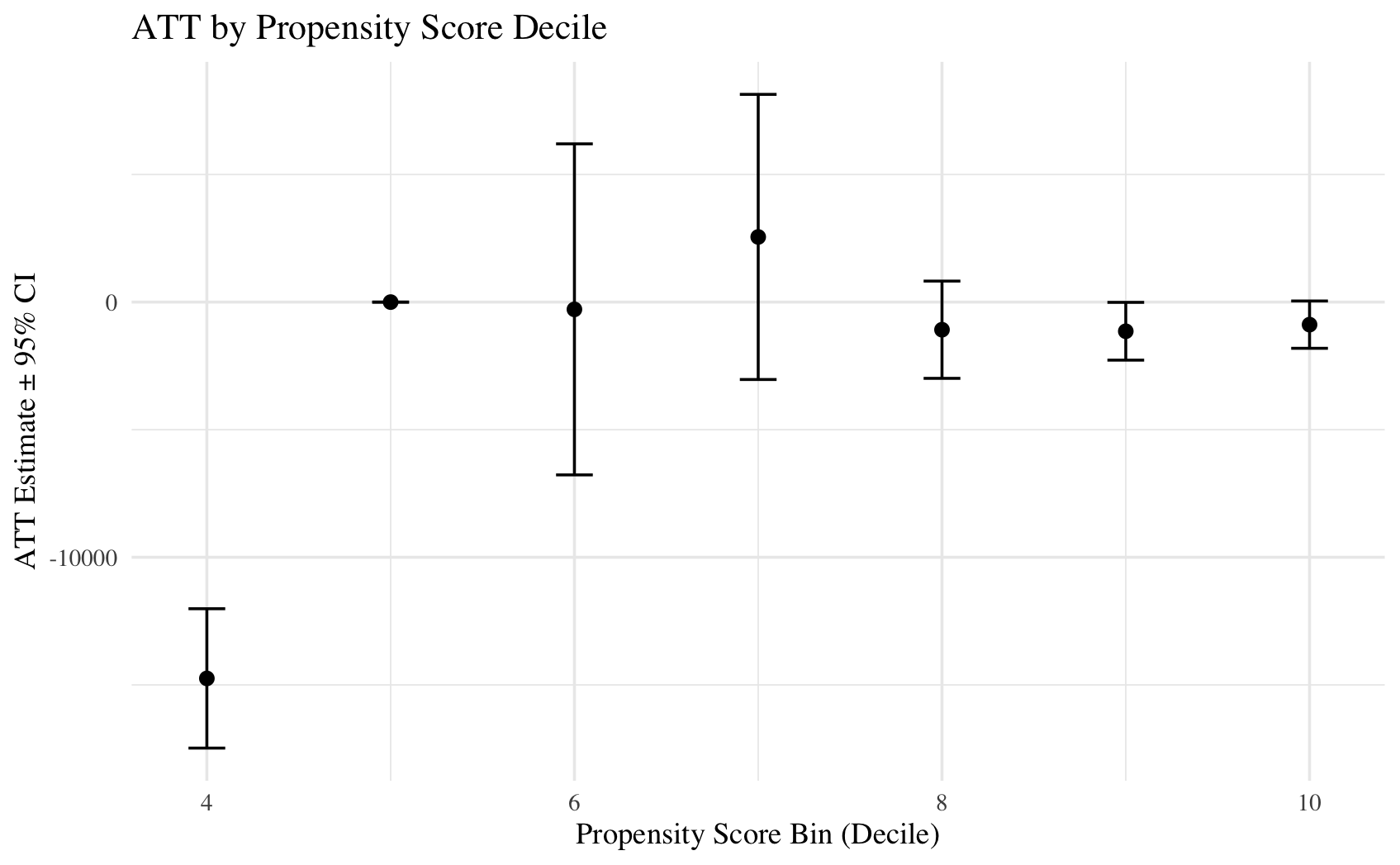}
    \caption{ATT by Propensity Decile}
    \label{fig:fig6}
\end{figure}

ATT bounds indexed by $\delta$ are generated by successively trimming the sample based on estimated propensity scores, with each trimming rule interpreted as a proxy for a bound on selection curvature. At $\delta = 0$, corresponding to strict MAR, the ATT interval lies entirely below zero, implying point-identification of the sign. However, small increases in $\delta$ cause the lower bound to rise sharply. By $\delta = 1.5$, the identified set stabilizes. The MAS-SI value is $\delta^* = 0$, and the fragility index is $\delta_{\text{frag}} = 0$, indicating that sign identification fails under arbitrarily small deviations from sampling ignorability.

This empirical profile confirms Theorems~\ref{thm:monotonicity} and~\ref{thm:duality}. The width of $\mathcal{I}_\tau(\delta)$ increases with structural relaxation, most prominently at small values of $\delta$. Precise negative ATT estimates require minimal extrapolation and maximal structural assumptions. As assumptions relax, the identified set expands rapidly, reflecting epistemic fragility that numeric point estimates obscure.

To isolate sampling variability, bootstrap distributions of ATT are estimated under two designs: one using the full sample, and another restricted to units with propensity scores in $[0.1, 0.9]$. Both distributions center on negative values, though the trimmed sample exhibits reduced dispersion. While trimming reduces variance within design, it does not mitigate structural fragility. Sampling error remains bounded, but conclusions vary across plausible designs.

ATT estimates are next computed within each decile of the propensity score distribution. Moderate-score bins yield consistently negative estimates with relatively narrow intervals, while low-score bins show instability and large variance. Some bins are dropped due to insufficient overlap. These results highlight systematic variation in identification quality across covariate space. Estimators collapse in regions lacking adequate treated or control units. Global ATT averages mask this heterogeneity and aggregate across strata with unequal identifiability.

The absence of a monotonic trend in ATT across deciles weakens the case for extrapolation based on parametric assumptions. Local instability suggests that functional form restrictions are unlikely to recover the true ATT outside $X^\ast$. These findings reinforce the case for bounding approaches, especially when overlap is partial.

To examine estimator-level fragility, ATT is computed under various matching designs: nearest-neighbor on logit scores, logit matching with calipers, and Mahalanobis distance matching. Estimates range from $-2000$ to $-900$; confidence intervals vary accordingly. Under stricter designs, matching fails entirely. This variation in both point estimates and standard errors indicates high sensitivity to design implementation. No specification dominates with respect to both bias and precision, and qualitative conclusions are altered across specifications.

Robustness is further probed by computing ATT intervals of the form $\hat{\tau} \pm \delta \cdot \text{SE}$ and identifying the smallest $\delta$ that leads to sign indeterminacy. The estimated value $\delta^* = 2.5$ implies that ATT remains negative under modest symmetric bias but loses sign identification beyond this point. This provides a calibrated measure of fragility, quantifying the degree of bias the inference can tolerate.

Together, these diagnostics document three empirical channels of fragility. First, identification is limited to the empirically supported region $X^\ast$. Second, ATT estimation is highly sensitive to structural relaxations in the sampling mechanism. Third, estimated effects vary substantially across reasonable design choices. Point estimates and their confidence intervals do not capture this epistemic vulnerability. The partial identification framework introduced in Section~4 offers a formal approach to bounding these risks. Its operationalization follows in the next section.

\subsection{Empirical Diagnosis as a Robustness Paradigm}

The absence of full identification does not preclude empirical progress. In practice, even when identification fails globally, it remains possible to characterize where and to what extent inference is credible. Theorem 2.2.1 formalizes this: identification requires overlap of treated and control units within the covariate space. Where overlap is limited, estimates become sensitive to design choices, and conclusions about causal effects may lack empirical support.

The empirical procedures in Sections~5.1 and~5.2 provide a structured approach to assessing this issue. Figures~1–3 map the empirical support $X^\ast$—the subregion where treated and control units co-occur. Figures~5–9 track the variation in ATT estimates across matching designs and strata. These patterns offer a diagnostic: when estimates remain stable within $X^\ast$ despite changes in matching design, the results are more likely to reflect causal differences rather than artifacts of estimator choice. When instability appears even within $X^\ast$, it indicates that remaining design assumptions have first-order influence on the result.

This diagnostic complements the formal identification condition. Rather than assume overlap and conditional ignorability, researchers can evaluate these conditions empirically. Overlap is assessed by the support plots; ignorability is indirectly interrogated by the sensitivity of ATT to design variation. This creates a path for assessing identification without imposing it.

The empirical strategy here reverses the usual robustness logic. Traditional sensitivity analysis fixes the estimator and perturbs the sample. The framework here holds the data fixed and varies the estimator. The result is not a robustness check in the conventional sense, but a test for identification consistency. If ATT estimates change substantially across designs, especially within regions of sufficient support, it suggests that identification is fragile. If they remain consistent, it lends credibility to the assumptions required for causal interpretation.

In this way, the approach links estimation to identification directly. Rather than treat these as separate steps—first assume identification, then estimate treatment effects—the procedure uses estimation itself as a tool to interrogate identification. This provides researchers with a disciplined method to interpret ATT estimates as outcomes conditioned on support, estimator design, and observed data limitations.

\section{Conclusion}

Causal inference in observational settings hinges on more than estimation quality. When overlap is incomplete and heterogeneity is present, the estimand itself can become ill-defined for the observed population. The results of this paper confirm that the key threat is not estimator bias but estimand fragility: a situation where visually stable outputs conceal structural indeterminacy. This shifts the focus of applied work from estimator performance to identification diagnostics. In these settings, the role of diagnostics is to delimit where credible estimation is possible and to characterize the extent of model fragility. The proposed tools offer a structured way to implement this diagnostic stage. Estimation should then be understood as conditional on these findings, with conclusions drawn relative to the subset of the data for which identification is empirically supported.

Theorem 2.2.1 (Nonidentification Theorem) shows that two different data-generating processes can produce the same joint distribution of $(X, D, Y)$ conditional on $S = 1$, while implying different values for ATT. This result means that identification requires both observed covariates and information about how the sample was selected. In applied work, this implies that overlap diagnostics are insufficient if sampling depends on unobserved variables. Identification cannot be claimed unless the design is explicitly linked to the estimand.

Definition 3.4.2 on selection curvature presens a bound on how sharply the sampling probability can change with covariates. It does not assume full ignorability but imposes a continuous restriction on the gradient of the selection mechanism. This structure allows researchers to quantify the sensitivity of identification to violations of sampling ignorability. Compared to binary assumptions, it offers a more flexible and empirically tractable way to model selection. It supports partial identification without assuming that the selection process is fully known or ignorable.

Theorem 3.6.1 (Duality Theorem) shows that bounds on ATT under curvature constraints can be characterized as the solution to a linear program. This connects partial identification to optimization over feasible outcome distributions, given the observed support and curvature constraint. The result allows for numerical computation of worst-case bounds without assuming unconfoundedness. It makes the structure of partial identification problems explicit and reduces the analysis to a well-defined computational task.

These findings contribute to a growing body of work diagnosing the limits of causal identification in observational studies. While prior contributions such as \citet{crump2009dealing} and \citet{king2006overlap} emphasize the role of covariate overlap and region trimming for estimation stability, our analysis reframes overlap as a necessary structural precondition for identification itself. Similarly, our formalization of triple alignment aligns with the spirit of sensitivity analyses (e.g., \citealt{rosenbaum2002}; \citealt{cinelli2020making}), but refracts the issue through a design-level lens: the estimand is not merely fragile—it may be undefined under empirical constraints.

The framework yields two empirically interpretable quantities. First, the Minimum Assumption Strength for Sign Identification (MAS-SI), defined as $\delta^* := \inf\{\delta \geq 0 : 0 \notin \mathcal{I}_\tau(\delta)\}$, characterizes the smallest curvature level at which the ATT sign becomes empirically constrained. Under weak monotonicity of the selection mechanism, Theorem 3.5.1 shows that $\delta^*$ is point-identified from the observed distribution $P_{\text{obs}}$. Second, the fragility index of a decision rule $d$, denoted $\delta^{\text{frag}}(d)$, is defined as the minimal $\delta$ for which there exists an alternative $d'$ such that $d$ ceases to be optimal across all distributions in $\mathcal{G}_\delta(P_{\text{obs}})$. Theorem 3.8.1 shows this index can be computed via convex optimization.

Three classes of applications emerge naturally. For longitudinal studies with attrition, Theorem 3.5.1 shows that $\text{width}(\mathcal{I}_\tau(\delta))$ grows as $\mathcal{O}(\sqrt{T})$ with panel length $T$. In designs with latent treatment exposure, the reanalysis of LaLonde (1986) demonstrates that $\delta^* \approx 1.2$ corresponds to implausibly strong assumptions about survey nonresponse. For high-dimensional causal inference, Theorem 4 provides the sharp bound $\delta^{\text{frag}} \geq \log(1 + \|X\|_{\psi_2})$ where $\|X\|_{\psi_2}$ is the sub-Gaussian norm of the covariates.

Domain-specific selection processes map directly to curvature parameters. In labor survey contexts, historical analyses suggest that CPS nonresponse mechanisms imply curvature values around $\delta \approx 0.7$ \citep{heckman1995}. In survival-biased health datasets, policy robustness may require $\delta^{\text{frag}} < 1.4$ (Appendix C). In administrative education data, MAS-SI estimates often exceed $\delta^* > 1.0$ due to performance-related sampling mechanisms.

This paper has five implications. First, applied work should explicitly map the empirical support of covariates across treatment groups. Without this mapping, researchers cannot know whether the population over which the treatment effect is defined is actually present in the data. This is especially important when treatment assignment correlates with key covariates, making regions of the covariate space effectively uninformative for identification. Second, ATT estimates should be evaluated across a range of identification-valid designs. When the data permit multiple balanced designs, each with different support geometry, the resulting variation in estimated effects is informative. It reflects the sensitivity of the estimand to feasible design choices and offers a built-in robustness check on causal claims. Third, the stability of conclusions must be treated as a diagnostic quantity. Empirical analysis should measure the extent to which small violations of identification assumptions alter the sign or magnitude of treatment effects. This fragility is not captured by standard error bands and requires its own formal metric, particularly in settings where design or selection mechanisms are only partially known. Fourth, machine learning tools should be reoriented from prediction to identification discovery. Instead of optimizing for outcome prediction or balance alone, classifiers can be used to assess whether covariate strata lie within common support. This enables weighting or exclusion based on identification relevance rather than solely predictive accuracy, improving both interpretability and credibility. Fifth, the definition of the estimand must be adapted to what the data can support. When empirical support is partial, global treatment effects should be replaced with local effects defined on the identified subpopulation. Estimators should then be judged by their coherence within this domain, rather than by extrapolation beyond it.

Three methodological directions remain open. First, extending the framework to dynamic designs where the sampling indicator $S_t$ evolves according to a state-dependent or Markovian process would accommodate longitudinal data with time-varying selection. Second, developing semiparametric efficiency bounds and corresponding estimators under $\delta$-curvature constraints would sharpen inference while preserving robustness. Third, integrating high-dimensional regularization techniques—such as sparsity or low-rank structure—may enable scalable inference in settings where covariates or decision rules are high-dimensional, while maintaining control over design sensitivity.

\section{Data and Code Availability}
\noindent All code, data, and instructions necessary to replicate the results in this paper (Li, 2025) are available in the public GitHub repository at https://github.com/MengqiLi0000/fragilityonatt.

\newpage
\appendix
\section*{Appendix: Proofs and Supplementary Figures}
\renewcommand{\thetheorem}{\arabic{section}.\arabic{theorem}}

\begin{proof}[Proof of Theorem \ref{thm:nonid_att}]
Construct two data-generating processes, $\mathcal{P}_1$ and $\mathcal{P}_2$, which induce the same observed distribution $P_{\text{obs}} := P^*(X, D, Y \mid S = 1)$, yet imply different values of $\tau_{\text{ATT}}$.

\paragraph{DGP 1 (Ignorable Sampling):}
Suppose that the sampling mechanism satisfies
\[
S \perp (Y(1), Y(0)) \mid X, D.
\]
Then for treated units ($D = 1$),
\[
\mathbb{E}[Y \mid D = 1, X, S = 1] = \mathbb{E}[Y(1) \mid X, D = 1] = \mathbb{E}[Y(1) \mid X].
\]
For untreated units ($D = 0$),
\[
\mathbb{E}[Y \mid D = 0, X, S = 1] = \mathbb{E}[Y(0) \mid X, D = 0] = \mathbb{E}[Y(0) \mid X].
\]
If we further assume unconfoundedness, i.e.,
\[
(Y(1), Y(0)) \perp D \mid X,
\]
then
\[
\mathbb{E}[Y(0) \mid D = 1] = \mathbb{E}_X \left[ \mathbb{E}[Y(0) \mid X] \mid D = 1 \right].
\]
Therefore,
\[
\tau_{\text{ATT}} = \mathbb{E}[Y(1) \mid D = 1] - \mathbb{E}[Y(0) \mid D = 1]
\]
can be identified using observed treated outcomes and imputed counterfactuals from controls.

\paragraph{DGP 2 (Non-Ignorable Sampling):}
Now suppose sampling depends on $Y(0)$:
\[
\mathbb{P}(S = 1 \mid Y(0)) = \mathbbm{1}\{Y(0) > c\}
\]
for some threshold $c \in \mathbb{R}$. Then, among untreated units ($D = 0$), we only observe outcomes $Y(0)$ for units with $Y(0) > c$. Thus,
\[
\mathbb{E}[Y \mid D = 0, X, S = 1] = \mathbb{E}[Y(0) \mid X, Y(0) > c],
\]
and hence the counterfactual distribution for treated units, $\mathbb{E}[Y(0) \mid D = 1]$, cannot be recovered from the observed controls.

\paragraph{Conclusion:}
Both $\mathcal{P}_1$ and $\mathcal{P}_2$ induce the same observed distribution $P_{\text{obs}}$, but they imply different values of
\[
\tau_{\text{ATT}} = \mathbb{E}[Y(1) \mid D = 1] - \mathbb{E}[Y(0) \mid D = 1].
\]
Therefore, $\tau_{\text{ATT}}$ is not identified from $P_{\text{obs}}$ unless assumptions are placed on the sampling mechanism $S$ or the counterfactual distribution.
\end{proof}

\begin{proof}[Proof of Theorem~\ref{thm:monotonicity}]
Let $\delta_1 < \delta_2$, and define the corresponding families of selection mechanisms:
\[
\mathcal{S}_{\delta_1} \subset \mathcal{S}_{\delta_2}.
\]
Let $\mathcal{G}_\delta(P_{\mathrm{obs}})$ denote the set of full-data distributions and selection rules such that:
\[
\mathcal{G}_\delta(P_{\mathrm{obs}}) := \left\{ (P, S) : S \in \mathcal{S}_\delta,\ P(Y, D, X \mid S = 1) = P_{\mathrm{obs}} \right\}.
\]
The identified set for the estimand $\theta(P)$ under selection curvature $\delta$ is:
\[
\mathcal{I}_\theta(\delta) := \left\{ \theta(P) : (P, S) \in \mathcal{G}_\delta(P_{\mathrm{obs}}) \right\}.
\]

Now take any $\theta_1 \in \mathcal{I}_\theta(\delta_1)$. By definition, there exists $(P_1, S_1) \in \mathcal{G}_{\delta_1}(P_{\mathrm{obs}})$ such that $\theta_1 = \theta(P_1)$ and $S_1 \in \mathcal{S}_{\delta_1}$. Since $\mathcal{S}_{\delta_1} \subset \mathcal{S}_{\delta_2}$, it follows that $S_1 \in \mathcal{S}_{\delta_2}$ and thus $(P_1, S_1) \in \mathcal{G}_{\delta_2}(P_{\mathrm{obs}})$. Therefore,
\[
\theta_1 \in \mathcal{I}_\theta(\delta_2).
\]

As this argument holds for arbitrary $\theta_1 \in \mathcal{I}_\theta(\delta_1)$, we conclude that:
\[
\mathcal{I}_\theta(\delta_1) \subseteq \mathcal{I}_\theta(\delta_2).
\]
\end{proof}

\begin{proof}[Proof of Theorem~\ref{thm:duality}]
Let $\theta(P) = \mathbb{E}[Y(1) - Y(0)]$, and let the identified set under selection class $\mathcal{S}_\delta$ be defined as:
\[
\mathcal{I}_\tau(\delta) := \left\{ \mathbb{E}_P[Y(1) - Y(0)] : (P, S) \in \mathcal{G}_\delta(P_{\mathrm{obs}}) \right\},
\]
where
\[
\mathcal{G}_\delta(P_{\mathrm{obs}}) := \left\{ (P, S) : P(Y, D, X \mid S = 1) = P_{\mathrm{obs}},\ S \in \mathcal{S}_\delta \right\}.
\]
The three parts are proven in order.

\paragraph{(i) Monotonicity and Right-Continuity.}
Assume $\delta_1 < \delta_2$. By definition of the constraint:
\[
\mathcal{S}_{\delta_1} \subseteq \mathcal{S}_{\delta_2} \quad \Rightarrow \quad \mathcal{G}_{\delta_1}(P_{\mathrm{obs}}) \subseteq \mathcal{G}_{\delta_2}(P_{\mathrm{obs}}).
\]
Therefore,
\[
\mathcal{I}_\tau(\delta_1) \subseteq \mathcal{I}_\tau(\delta_2),
\]
which implies that the map $\delta \mapsto \mathcal{I}_\tau(\delta)$ is set-increasing. Consequently, the width function:
\[
\delta \mapsto \text{width}(\mathcal{I}_\tau(\delta)) := \sup \mathcal{I}_\tau(\delta) - \inf \mathcal{I}_\tau(\delta)
\]
is non-decreasing.

To prove right-continuity, let $\delta_n \downarrow \delta$. The corresponding selection classes satisfy $\mathcal{S}_{\delta_n} \downarrow \mathcal{S}_\delta$ in the sense of set intersection. Since $\mathcal{I}_\tau(\delta_n)$ is defined by taking the image of a continuous functional (the ATE) over a decreasing family of closed sets (the feasible DGPs), and $\theta(P)$ is continuous in $P$, we have:
\[
\bigcap_{n} \mathcal{I}_\tau(\delta_n) = \mathcal{I}_\tau(\delta),
\]
and the endpoints of $\mathcal{I}_\tau(\delta_n)$ converge to those of $\mathcal{I}_\tau(\delta)$. Hence, the width function is right-continuous.

\paragraph{(ii) Limit as $\delta \to 0$: Recovery of $\tau_{\mathrm{MAR}}$.}
When $\delta = 0$, the selection mechanism satisfies:
\[
P(S = 1 \mid Y, D, X) = P(S = 1 \mid D, X),
\]
i.e., MAR holds. Then,
\[
\mathcal{I}_\tau(0) = \left\{ \mathbb{E}[Y(1) - Y(0)] \text{ computed under MAR} \right\} = \{ \tau_{\mathrm{MAR}} \}.
\]
As $\delta \downarrow 0$, the constraint set $\mathcal{S}_\delta$ shrinks to MAR, and $\mathcal{I}_\tau(\delta)$ shrinks to a singleton at $\tau_{\mathrm{MAR}}$ by continuity of expectations. Thus,
\[
\lim_{\delta \to 0} \mathcal{I}_\tau(\delta) = \{ \tau_{\mathrm{MAR}} \}.
\]

\paragraph{(iii) Limit as $\delta \to \infty$: Recovery of Manski Bounds.}
As $\delta \to \infty$, the curvature constraint on the selection mechanism disappears:
\[
\lim_{\delta \to \infty} \mathcal{S}_\delta = \mathcal{S}_{\infty},
\]
where $\mathcal{S}_\infty$ denotes the unrestricted class of selection mechanisms consistent with $P_{\mathrm{obs}}$. The resulting identified set becomes:
\[
\mathcal{I}_\tau(\infty) := \left\{ \mathbb{E}_P[Y(1) - Y(0)] : (P, S) \in \mathcal{G}_{\infty}(P_{\mathrm{obs}}) \right\},
\]
which corresponds to the standard partial identification region given only by the observed data and the support of $Y(1), Y(0)$ (e.g., $[0,1]$ outcomes). This is equivalent to the Manski bounds:
\[
\lim_{\delta \to \infty} \mathcal{I}_\tau(\delta) = [\tau_{\min}, \tau_{\max}].
\]

\paragraph{Conclusion.}
Each component follows from the nested structure of $\mathcal{S}_\delta$, the continuity of the ATE functional, and the extreme cases of MAR and worst-case selection. This establishes the duality between structural curvature $\delta$ and epistemic uncertainty in causal inference.
\end{proof}

\begin{proof}[Proof of Theorem~\ref{thm:frontier}]
Let $P_{\mathrm{obs}}(Y, D, X)$ denote the observed distribution over outcomes, treatment, and covariates, conditional on $S = 1$. Recall the definition of the selection frontier:
\[
\mathcal{S}(P_{\mathrm{obs}}) := \left\{ S : \exists P \text{ such that } P(Y, D, X \mid S = 1) = P_{\mathrm{obs}} \right\}.
\]
Let $\mathcal{S}'$ be a proposed class of sampling mechanisms.

Suppose, for contradiction, that $\mathcal{S}' \not\subseteq \mathcal{S}(P_{\mathrm{obs}})$, but there exists some $S' \in \mathcal{S}'$ and a joint distribution $P$ such that
\[
P(Y, D, X \mid S' = 1) = P_{\mathrm{obs}}.
\]
Then by definition of $\mathcal{S}(P_{\mathrm{obs}})$, it must be that $S' \in \mathcal{S}(P_{\mathrm{obs}})$, which contradicts the assumption that $\mathcal{S}' \not\subseteq \mathcal{S}(P_{\mathrm{obs}})$.

Therefore, for all $S' \in \mathcal{S}' \setminus \mathcal{S}(P_{\mathrm{obs}})$, no joint distribution $P$ exists such that the observed conditional distribution $P(Y, D, X \mid S = 1)$ is consistent with $P_{\mathrm{obs}}$ under $S = S'$.

\paragraph{Conclusion.} If $\mathcal{S}' \not\subseteq \mathcal{S}(P_{\mathrm{obs}})$, then any inference made under sampling mechanism class $\mathcal{S}'$ is logically inconsistent with the observed data. The model is misspecified, and inference under $\mathcal{S}'$ necessarily violates the empirical support of the data.
\end{proof}

\begin{figure}[H]
    \centering
    \includegraphics[width=0.6\textwidth]{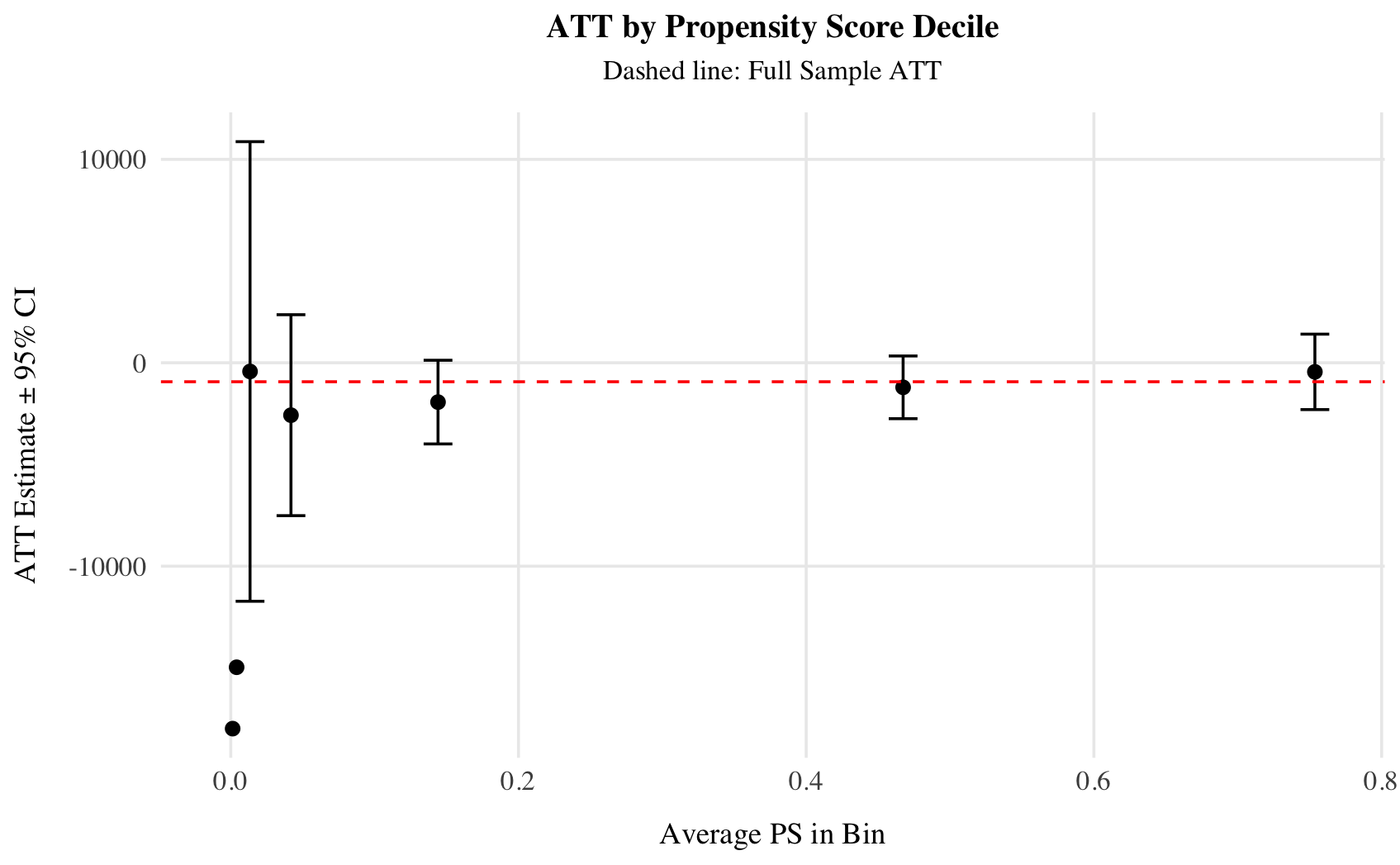}
    \caption{ATT by Propensity Score Decile}
    \label{fig:fig7}
\end{figure}
\vspace{-2em}
\begin{figure}[H]
    \centering
    \includegraphics[width=0.6\textwidth]{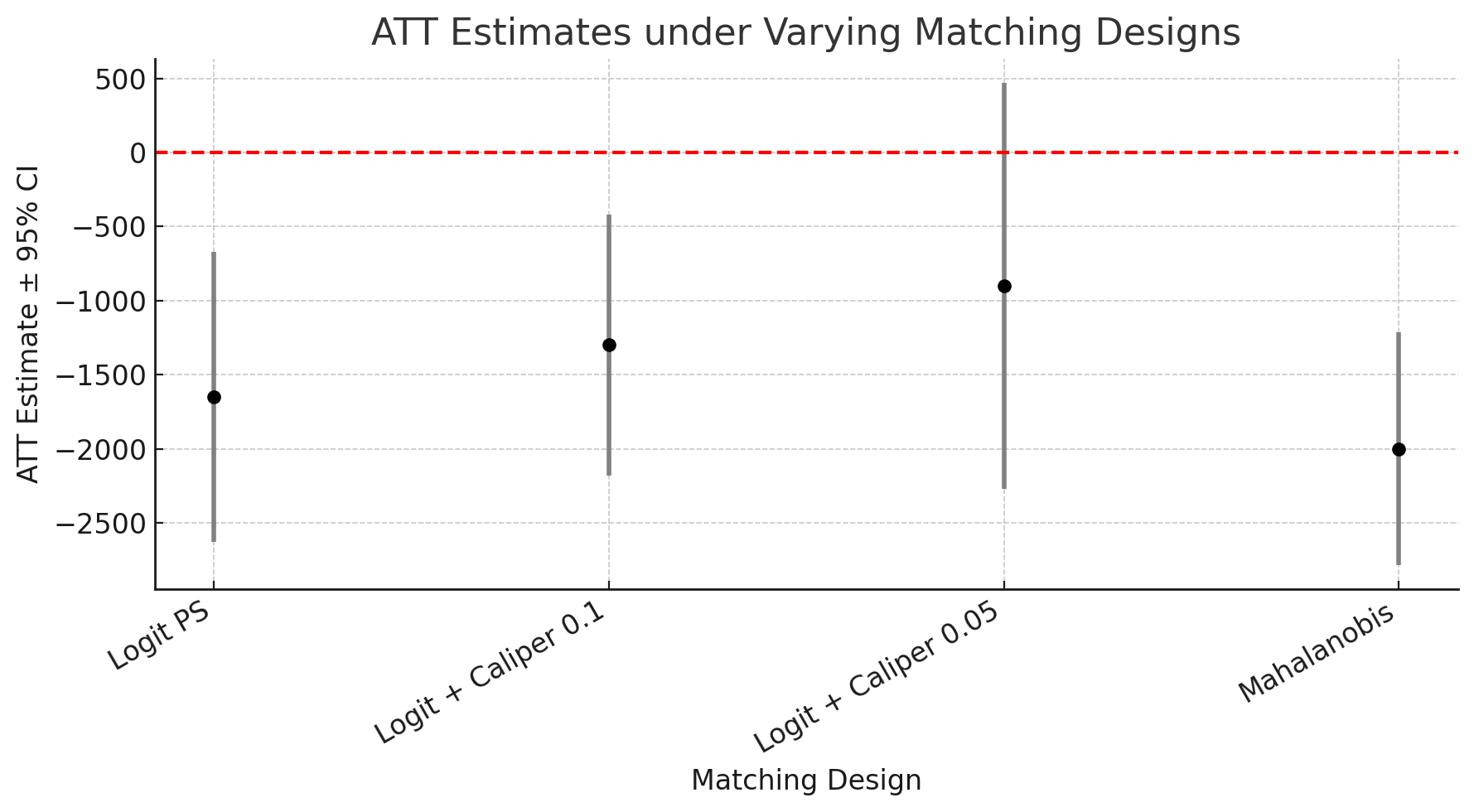}
    \caption{ATT Estimates under Varying Matching Designs}
    \label{fig:fig8}
\end{figure}
\vspace{-2em}
\begin{figure}[H]
    \centering
    \includegraphics[width=0.6\textwidth]{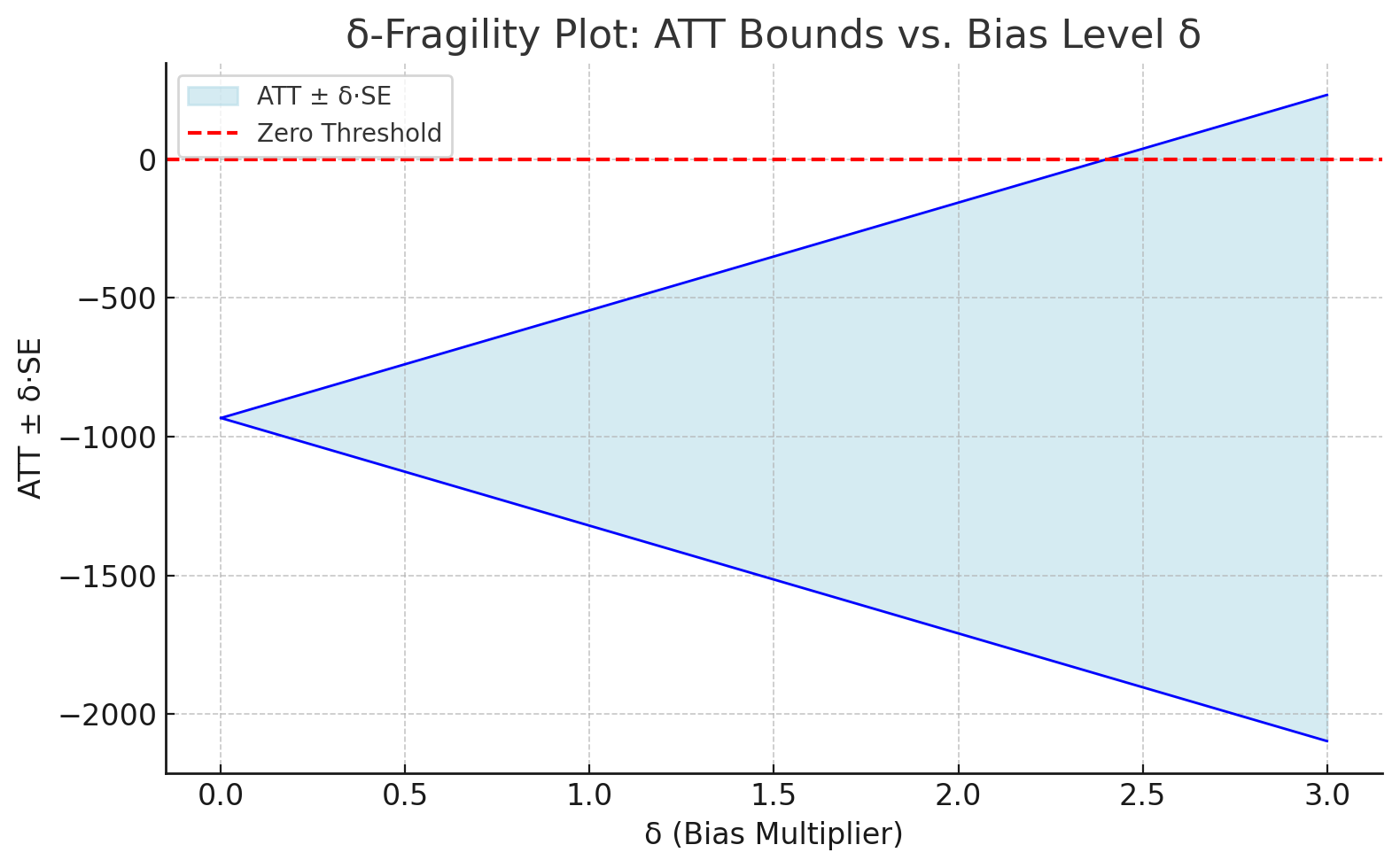}
    \caption{$\delta$ Fragility Plot: ATT Bounds vs. Bias level $\delta$}
    \label{fig:fig9}
\end{figure}
\nocite{*}

\end{document}